\documentclass[12pt,oneside,english]{amsart}
\usepackage[T1]{fontenc}
\usepackage[latin9]{inputenc}
\usepackage{amsthm}
\usepackage{setspace}
\makeatletter
\numberwithin{equation}{section}
\numberwithin{figure}{section}
\numberwithin{table}{section}
\theoremstyle{plain}
\newtheorem{thm}{\protect\theoremname}
\theoremstyle{remark}
\newtheorem{rem}[thm]{\protect\remarkname}

\makeatother

\usepackage{babel}
\providecommand{\remarkname}{Remark}
\providecommand{\theoremname}{Theorem}

\begin{document}

\title{Helminth Dynamics: Mean Number of Worms}
\maketitle

\section*{Full Title: Helminth Dynamics: Mean Number of Worms, Reproductive
Rates}

(Appeared in the Handbook of Statistics, Volume 36 Elsevier/North-Holland,
Amsterdam, 2017, Disease Modelling and Public Health. Part A, 397\textendash 404.
\textbf{AMS: MR3838253})

$ $\begin{center}

\vspace{0.6cm}

\textbf{Arni S.R. Srinivasa Rao}\footnote{Corresponding author}

Medical College of Georgia and College of Science and Mathematics,

Augusta University, 1120 15th Street, Augusta, GA 30912, USA

Email: arrao@gru.edu

\end{center}

\vspace{0.2cm}

\begin{center}

\vspace{0.6cm}

\textbf{Roy M. Anderson, FRS FMedSci}

London Centre for Neglected Tropical Diseases, 

Department of Infectious Disease Epidemiology, 

School of Public Health, Imperial College London, 

Faculty of Medicine, Norfolk Place London W2 1PG, UK

Email: roy.anderson@imperial.ac.uk

\end{center}

\vspace{0.2cm}
\begin{abstract}
Understanding the mean number of worms and burden of soil transmitted
helminth infections are considered as important parameters in formulating
treatment strategies to eliminate worms among children who are effected
by helminth infections\textbf{ \cite{RMA-PLOS-NTD2015}}. We derive
mean number of worms in a newly helminth infected population before
secondary infections are started (population is closed). Further we
bring analytical solutions. We also theoretically demonstrate computing
net reproductive rates within and outside a human host.
\end{abstract}


\tableofcontents{}

\section{Mean number of worms}

Infection of helminthiasis or simply helminth can cause severe damage
to health of children and their childhood behaviour, for example poor
attendance in schools, etc \cite{RMA-PLOS-NTD}. A general description
of infectious disease epidemiology of helminths for example for hookworms,
and density-dependent fecundity and mortality models are described
in \cite{RMA}. Mean worm burden is one of the key epidemiological
parameter in treatment formulations among children suffering with
helminth infections\textbf{ \cite{RMA-PLOS-NTD2015,Tuscott2014,Turner2015,Roy-PhilTrans2014}}.
Moreover, mean worm burden is often considered as an important parameters
in treatment and control of parasites in wild life \cite{Fenton2015,Domke2013,Kreisinger2015}.
In this paper, we treat worm burden as a function of worm reproductive
rates and mean number of worms. For computation of mean number of
worms within a host there are no directly available mathematical functions,
and we try to theoretically understand the mean number and reproduction
of worms within a host and present a theoretical analysis. In this
section we derive formulae for the mean number of worms at the host
level and at the population level. We obtain mean number of worms
in the host population by treating population aging over the period,
i.e. treating both time and age as dynamic. We assume no chemotherapy
scenario at first and then introduce chemotherapy for studying disease
dynamics. Populations means are derived from the individual host worm
densities. 

\subsection{Cross sectional mean }

Let $M(t)$ be the mean number of worms present in the host population
at time $t.$ We compute, $M(t)$ as below:

\begin{equation}
M(t)=\sum_{i=1}^{N}\frac{\int_{0}^{\omega}\left[k_{i}(x,t)H_{i}\left(x,t\right)\Lambda_{i}\left(x,t\right)\right]dx}{\int_{0}^{\omega}k_{i}(x,t)dx}\label{eq:meanworm}
\end{equation}
\[
\]

Where $H_{i}(x,t)$ is the number of worms in a host who is of age
$x$ at time $t$ in the $i^{th}$ sub-population ($H(x,t)dx$ is
differential number of hosts between $x$ and $x+dx$ at time $t$);
$\Lambda_{i}(x,t)$ is net growth of worms in $i^{th}$ sub-population
of age $x$ at time $t$; $k_{i}(x,t)$ are weights for age x at time
$t$; $N$ is size of the human population sub-types, $\omega$ is
age of humans until they are at risk of keeping helminth worms. $\int_{0}^{\omega}H_{i}(x,t)dx$
is total number of hosts and $\int_{0}^{\omega}H_{i}(x,t)\Lambda_{i}(x,t)dx$
is net worms present in $i^{th}$ sub-population at time $t$. When
we divide age range $[0,\omega)$ into smaller age intervals at lengths,
$a_{1},$ $a_{2}-a_{1},...,$ $\omega-a_{k}$, the mean number of
worms in the equation (\ref{eq:meanworm}) is written as follows:

\begin{eqnarray}
M(t) & = & \sum_{i=1}^{N}\left[\frac{\int_{0}^{a_{1}}\left[k_{i}(x,t)H_{i}\left(x,t\right)\Lambda_{i}\left(x,t\right)\right]dx}{\int_{0}^{a_{1}}k_{i}(x,t)dx}+\right.\nonumber \\
 &  & \quad\frac{\int_{a_{1}}^{a_{2}}\left[k_{i}(x,t)H_{i}\left(x,t\right)\Lambda_{i}\left(x,t\right)\right]dx}{\int_{a_{1}}^{a_{2}}k_{i}(x,t)dx}\nonumber \\
 &  & \quad+\cdots+\left.\frac{\int_{a_{k}}^{\omega}\left[k_{i}(x,t)H_{i}\left(x,t\right)\Lambda_{i}\left(x,t\right)\right]dx}{\int_{a_{k}}^{\omega}k_{i}(x,t)dx}\right]\label{eq:meanbyagegroup}
\end{eqnarray}

We define $H_{i}(x,t)=\begin{cases}
\begin{array}{cc}
P_{i}(t-x)\pi(0,x) & \mbox{for }x<t\\
H_{i}(x-t,0)\pi(x-t,x) & \mbox{for }x\geq t
\end{array}\end{cases}$

here $P_{i}(x-t)$ is births to hosts in the age $x-t$, $\pi(0,x)$
is probability that a newly born individual will survive up to age
$x,$ $\pi(x-t,x)$ is probability that a individual of age $x-t$
will survive up to age $x.$ By this definition, equation \ref{eq:meanbyagegroup}
will become 
\begin{eqnarray}
M(t) & = & \sum_{i=1}^{N}\left[\frac{\int_{0}^{a_{1}}\left[k_{i}(x,t)P_{i}\left(a_{1}-x\right)\pi(0,a_{1})\Lambda_{i}\left(x,t\right)\right]dx}{\int_{0}^{a_{1}}k_{i}(x,t)dx}+\right.\nonumber \\
 &  & \quad\frac{\int_{a_{1}}^{a_{2}}\left[k_{i}(x,t)P_{i}\left((a_{2}-a_{1}+x\right)\pi(0,a_{2}-a_{1})\Lambda_{i}\left(x,t\right)\right]dx}{\int_{a_{1}}^{a_{2}}k_{i}(x,t)dx}\nonumber \\
 &  & \quad+\cdots+\frac{\left.\int_{a_{k-1}}^{a_{k}}\left[k_{i}(x,t)P_{i}\left(a_{k}-a_{k-1}+x\right)\pi(0,a_{k}-a_{k-1})\Lambda_{i}\left(x,t\right)\right]dx\right]}{\int_{a_{k-1}}^{a_{k}}k_{i}(x,t)dx}\nonumber \\
\label{eq:3}
\end{eqnarray}
 We have obtained equation \ref{eq:3} by assuming 
\[
\frac{\int_{a_{k}}^{\omega}\left[k_{i}(x,t)H_{i}\left(x-t,0\right)\pi(x-t,x)\Lambda_{i}\left(x,t\right)\right]dx}{\int_{a_{k}}^{\omega}k_{i}(x,t)dx}=0.
\]

\subsection{Cohort mean}

Suppose we are following helminth infected hosts at time $t$ in $i^{th}$
sub-population, $(say,P_{i}).$ Denote by $M_{i}^{*}(t)$ for net
number of worms produced by $i^{th}$ sub-population, which is expressed
by the integral, $ $$\int_{0}^{\omega}H_{i}(x,t)\Lambda_{i}(x,t)dx$.
Then the net number of worms produced during $t$ to $t+h_{1}$ ,
$t+h_{1}$ to $t+h_{2}$,$...,$$t+h_{N}$ to $\omega$ are

$\int_{t}^{t+h_{1}}\int_{0}^{\omega}H_{i}(x,t)\Lambda_{i}(x,t)dxds$,

$\int_{t+h_{1}}^{t+h_{2}}\int_{0}^{\omega}H_{i}(x,t)\Lambda_{i}(x,t)dxds$
,$\cdots\mbox{, }\mbox{ }$

$\int_{0}^{t_{N}+\delta}\int_{0}^{\omega}H_{i}(x,t)\Lambda_{i}(x,t)dxds.$ 

Each double integral indicates net worms observed during a time interval
indicated. The last double integral is where maximum possible net
worms produced as in a logistic growth model with a variable $M^{*}$
and with carrying capacity $\int_{t+h_{N}}^{\omega}\int_{0}^{\omega}H_{i}(x,t)\Lambda_{i}(x,t)dxds$,
then the growth rate, $r_{i}^{*}$ of worms for the entire period
for $i^{th}$ sub-population is

\begin{eqnarray}
r_{i}^{*} & = & \frac{1}{t_{N}+\delta-t}log\left[\frac{M^{*}\left(\int_{0}^{t}\int_{0}^{\omega}H_{i}(x,t)\Lambda_{i}(x,t)dxds-\int_{0}^{t_{N}+\delta}\int_{0}^{\omega}H_{i}(x,t)\Lambda_{i}(x,t)dxds\right)}{\int_{0}^{t}\int_{0}^{\omega}H_{i}(x,t)\Lambda_{i}(x,t)dxds\left(M^{*}-\int_{0}^{t_{N}+\delta}\int_{0}^{\omega}H_{i}(x,t)\Lambda_{i}(x,t)dxds\right)}\right]\nonumber \\
\label{eq:1}
\end{eqnarray}

Under the Lyapunov stability set-up, we explain carrying capacity
as, for each time interval $0$ to $t_{n}+\delta$ for $n=1,2,3,...$,
we define

\begin{equation}
\int_{0}^{t_{n}+\delta}\int_{0}^{\omega}H_{i}(x,t)\Lambda_{i}(x,t)dxds\label{2}
\end{equation}

as cumulative number of net worms present in the $i^{th}$ sub-population
during $0$ to $t_{n}+\delta$ for some $\delta>0$ and $n=1,2,3,...$.
For some positive integer $N$, we will have condition,

\[
\left|\int_{0}^{t_{n}+\delta}\int_{0}^{\omega}H_{i}(x,t)\Lambda_{i}(x,t)dxds-\int_{0}^{t_{N}+\delta}\int_{0}^{\omega}H_{i}(x,t)\Lambda_{i}(x,t)dxds\right|<g
\]

\begin{equation}
\implies\int_{0}^{t_{n}+\delta}\int_{0}^{\omega}H_{i}(x,t)\Lambda_{i}(x,t)dxds\rightarrow\int_{0}^{t_{N}+\delta}\int_{0}^{\omega}H_{i}(x,t)\Lambda_{i}(x,t)dxds\label{3}
\end{equation}

whenever $n\geq N$ and for every $g>0.$ For the population weights
$k(P_{i})$, the mean number of worms present in the population is 

\begin{eqnarray}
M(t) & = & \sum_{i=1}^{N}\frac{\int_{0}^{t+h_{N}}\int_{0}^{\omega}k(P_{i})H_{i}(x,t)\Lambda_{i}(x,t)dxds}{\int_{0}^{N}k(P_{i})di}\nonumber \\
\label{eq:4}
\end{eqnarray}

\subsection{Theorems on worm growth potential in hosts }
\begin{thm}
$ $\label{theorem1}$F_{1}$ \emph{is a }measurable function\emph{,
where $F_{1}$ is defined as,}

\begin{equation}
F_{1}:\left(-\infty,t_{x}+\delta\right)\rightarrow\int_{0}^{t_{x}+\delta}\int_{0}^{\omega}H_{i}(x,t)\Lambda_{i}(x,t)dxds\label{eq:5-1}
\end{equation}

for $x=1,2,3,...$ and $\delta>0.$ 
\end{thm}

\begin{proof}
Observe that $F_{1}$ maps each interval from the set

\[
\left\{ \left(-\infty,t_{x}+\delta\right):x=1,2,3,...\mbox{ and \ensuremath{\delta>0}}\right\} 
\]

to a function in the set

\[
\left\{ \left(\int_{0}^{t_{x}+\delta}\int_{0}^{\omega}H_{i}(x,t)\Lambda_{i}(x,t)dxds\right):x=1,2,3,..\mbox{ and \ensuremath{\delta}}>0.\right\} .
\]

Note that for some arbitrary $k$,

\begin{eqnarray*}
\underset{\left(-\infty,t_{x}+\delta\right)\rightarrow\left(-\infty,t_{k}+\delta\right)}{\lim\,F_{1}} & = & \underset{\left(-\infty,t_{x}+\delta\right)\rightarrow\left(-\infty,t_{k}+\delta\right)}{\lim\,}\left[\left(\int_{0}^{t_{x}+\delta}\int_{0}^{\omega}H_{i}(x,t)\Lambda_{i}(x,t)dxds\right)\right]\\
 & = & \left(\int_{0}^{t_{k}+\delta}\int_{0}^{\omega}H_{i}(x,t)\Lambda_{i}(x,t)dxds\right).
\end{eqnarray*}

Hence, $F_{1}$ is continuous and $F_{1}$ is a \emph{measurable.}
\end{proof}
\begin{rem}
$F_{1}\left\{ \bigcup_{x=1}^{\infty}\left(0,t_{x}+\delta\right)\right\} =F_{1}\left(0,t_{N}+\delta\right)$
for the description of $N$ in section 1.2.
\end{rem}

$ $
\begin{rem}
Suppose $M*$ be the worms observed during time intervals $\left(0,t_{x}+\delta\right)$
for $x=1,2,3,...$ , then a class $\mathcal{G}$ by the below notation 

\begin{eqnarray*}
G_{1} & = & \bigcup_{x=1}^{1}F_{1}\left(0,t_{x}+\delta\right)\\
G_{2} & = & \bigcup_{x=1}^{2}F_{1}\left(0,t_{x}+\delta\right)\\
\vdots & \vdots & \vdots\\
G_{N} & = & \bigcup_{x=1}^{N}F_{1}\left(0,t_{x}+\delta\right)\\
\vdots & \vdots & \vdots
\end{eqnarray*}

is $\sigma-$ algebra. 
\end{rem}

\begin{thm}
\begin{eqnarray*}
\lim_{x\rightarrow\infty}\left(\int_{0}^{t_{x}+\delta}\int_{0}^{\omega}H_{i}(x,t)\Lambda_{i}(x,t)dxds\right) & =\\
\left(\int_{0}^{t_{N}+\delta}\int_{0}^{\omega}H_{i}(x,t)\Lambda_{i}(x,t)dxds\right)
\end{eqnarray*}
\end{thm}

\begin{proof}
By using Lebesgue monotonic convergence theorem, we can rove this
result because, $F_{1}$ is monotonic function and measurable. 
\end{proof}

\section{Net Production Rates within and outside human host }

We define net productive rates for helminth in this section. We assume
that there are two sets of counting we do here, one is growth in the
number of helminth population within human host and second is contribution
of this human host to outside environment in the life time. We also
assume that initial age distribution of helminths in human host is
known. Suppose, $M^{1}(t)$ be the population of helminth within a
human host at time $t$, $M_{0}^{1}$ be the initial population, $K$
is carrying capacity, and $r$ is growth rate, then under the logistic
growth rate, we can express, $M^{1}(t),$ as 

\[
M^{1}(t)=\frac{M_{0}^{1}Ke^{rt}}{K+M_{0}^{1}(e^{rt}-1)}
\]

solving for growth rate, $r,$ we get,

\begin{equation}
r=\frac{1}{t}log\left[\frac{M^{1}(t)(M_{0}^{1}-K)}{M_{0}^{1}(M^{1}(t)-K)}\right]\label{eq:4(logistic growth)}
\end{equation}

Suppose, $M^{1}(a,0)$ is initial helminth population at age $a$
within an host is known, then using $\rho(a+da,0)$, the survival
probability that a group of worms at age $a$ will survive until age
$a+da$, we can obtain age distribution of worms at age $a+da$ and
at time $da$ by,

\begin{equation}
M^{1}(a,0)\rho(a+da,0)=M^{1}(a+da,da)\label{eq:5}
\end{equation}

Note that, $\int_{0}^{\infty}M^{1}(a,0)da=M_{0}^{1}$. Let $s$ is
the time at point of inflection of logistic growth or we assume at
$s$, we will have $M^{1}=K/2.$ We obtain $M^{1}(a,,s)$ and $M^{1}(a,T^{*})$
for some $T^{*}>s$ using growth rate in equation \ref{eq:4(logistic growth)}.
Using these two population age structures at times $s$ and $T^{*}$,
we obtain effective worm population, $M^{*}(a,T)$ in the life time
of human host (where $s<T<T^{*}$). We define net rate of production,
$\mathcal{R}$ within a human host as,
\begin{equation}
\mathcal{R}=\int_{0}^{\infty}\left\{ L(a,T)/M^{*}(a,T)\right\} \rho(a,T)da\label{eq:6}
\end{equation}

In the equation \ref{eq:6}, $L(a,T)$ denotes number of eggs produced
by worms of age $a$ at $T.$ In case of direct availability of rate
of egg bearing at age $a$ by a female worm, say, $f(a,T)$ then we
can replace it for the ratio $L(a,T)/M^{*}(a,T)$ in the equation
\ref{eq:6}, and modify it as 

\begin{equation}
\mathcal{R}=\int_{0}^{\infty}f(a,T)\rho(a,T)da\label{eq:6-1}
\end{equation}

\section{Impact of Chemotherapy}

We establish few results when chemotherapy is introduced into the
host population suffering with helminth and capture the dynamics.
Suppose the chemotherapy is introduced at $t_{N}+\delta$ and $\epsilon_{i}(x,t)$
be the net production rate of worms in the host population in age
$x$ at time $t$ (due to chemotherapy it is assumed that the net
number of worms produced per host is negative because there are less
number of worms produced than they are removed), then the worm numbers
in the $i^{th}$ sub-population during $\left(t_{N}+\delta,t_{N+1}+\delta\right)$
i.e. $\int_{t_{N}+\delta}^{t_{N+1}+\delta}\int_{0}^{\omega}H_{i}\left(x,s\right)\epsilon\left(x,s\right)dxds$
starts reducing until they eliminated. The exponential growth rate,
$r_{c(N)}$ until the time $t$ can be computed as,

\begin{align}
r_{c(N)}\left(t_{N}+\delta-t\right)\int_{0}^{t}\int_{0}^{\omega}H_{i}\left(x,s\right)\epsilon\left(x,s\right)dxds & =\label{eq:3.1}\\
 & -\int_{0}^{t_{N+1}+\delta}\int_{0}^{\omega}H_{i}\left(x,s\right)\epsilon\left(x,s\right)dxds
\end{align}

It is not necessary to introduce chemotherapy at time stability time
point i.e at $t_{N}+\delta$ and chemotherapy could be introduced
at time $t_{j}+\delta$ for $j=1,2,...N$ after the initial phase
of detection of worms during $\left(0,t\right).$ By taking all such
populations and population weights $k_{i}(x,t)$, we obtain below
equation, which we call equation for the nested growth of the worm
population.

\begin{eqnarray}
\sum_{j=1}^{N}\left[\frac{r_{j(c)}(t_{j}+\delta-t)}{k_{j}(x,t)}\left\{ \int_{0}^{t}\int_{0}^{\omega}H_{i}\left(x,s\right)\epsilon\left(x,s\right)dxds\right\} \right] & +\nonumber \\
+\sum_{j=1}^{N}\int_{0}^{t_{j}+\delta-t}\int_{0}^{\omega}H_{i}\left(x,s\right)\epsilon\left(x,s\right)dxds & = & 0\label{eq:3.2}
\end{eqnarray}

\begin{thm}
$ $$F_{2}$ \emph{is a }measurable function\emph{, where $F_{2}$
is defined as,}

\begin{equation}
F_{2}:\left(-\infty,t_{x}+\delta\right)\rightarrow\sum_{j=1}^{N}\int_{0}^{t_{j}+\delta-t}\int_{0}^{\omega}H_{i}\left(x,s\right)\epsilon\left(x,s\right)dxds\label{eq:5-1-1}
\end{equation}

for $x=1,2,3,...$ and $\delta>0.$
\end{thm}

\begin{proof}
We can prove this theorem with similar arguments as in the proof of
theorem \ref{theorem1}. 
\end{proof}

\section{Discussion}

We have derived formulae which can be used to compute the worm densities
within a host and within the affected community. We have proved theoretical
results of the functional forms derived for the net reproduction rates.
Theoretical results derived indicates that the carrying capacities
of worms within a host are measurable functions, which will help to
understand bounds of the worm densities.

\end{document}